\documentclass[aps,prl,superscriptaddress,longbibliography,showpacs,twocolumn]{revtex4-2}

\usepackage{soul}
\setstcolor{red}
\usepackage{comment}

\usepackage{caption}
\DeclareCaptionJustification{justified}{\leftskip=0pt \rightskip=0pt \parfillskip=0pt plus 1fil}

\usepackage{graphicx}
\usepackage{dcolumn}
\usepackage{bm}
\usepackage{caption}
\usepackage{amsmath,amssymb,amsthm}
\usepackage{braket}
\usepackage{empheq}
\usepackage{subfig}
\usepackage{tikz}
\usepackage{url}
\usepackage{hyperref}
\usepackage{quantikz}
\usepackage{xcolor}
\usepackage{float}

\usepackage{algorithm}
\usepackage[noend]{algpseudocode}

\newcommand\norm[1]{\left\lVert#1\right\rVert}
\newcommand\dnorm[1]{\left\lVert#1\right\rVert_{\diamond}}
\newcommand\pnorm[1]{\left\lVert#1\right\rVert_{pauli}}

\hypersetup{
    breaklinks=true,
    colorlinks=true,       
    linkcolor=blue,          
    citecolor=red,        
    filecolor=magenta,      
    urlcolor=blue,           
    runcolor=cyan
}

\newtheorem{theorem}{Theorem}
\newtheorem*{theorem*}{Theorem}

\newtheorem{lemma}[theorem]{Lemma}

\newtheorem{definition}{Definition}

\newtheorem{remark}{Remark}
\begin{document}

\title{A quantum algorithm to simulate Lindblad master equations}
\author{Evan Borras}
\affiliation{Center for Quantum Information and Control, University of New
Mexico, Albuquerque, NM 87131, USA}
\affiliation{Department of Physics and Astronomy, University of New Mexico,
Albuquerque, NM 87131, USA}
\author{Milad Marvian}
\affiliation{Center for Quantum Information and Control, University of New
Mexico, Albuquerque, NM 87131, USA}
\affiliation{Department of Physics and Astronomy, University of New Mexico,
Albuquerque, NM 87131, USA}
\affiliation{Department of Electrical \& Computer Engineering, University of New Mexico, Albuquerque, NM 87131, USA}

\begin{abstract}
We present a quantum algorithm for simulating a family of Markovian master
equations that can be realized through a probabilistic application of unitary
channels and state preparation. Our approach employs a second-order product
formula for the Lindblad master equation, achieved by decomposing the dynamics
into dissipative and Hamiltonian components and replacing the dissipative
segments with randomly compiled, easily implementable elements.  The sampling
approach eliminates the need for ancillary qubits to simulate the dissipation
process and reduces the gate complexity in terms of the number of jump
operators. We provide a rigorous performance analysis of the algorithm. We also
extend the algorithm to time-dependent Lindblad equations, generalize the noise
model when there is access to limited ancillary systems, and explore
applications beyond the Markovian noise model. A new error bound, in terms of
the diamond norm, for second-order product formulas for
time-dependent Liouvillians is provided that might be of independent interest.

\end{abstract}

\maketitle

One of the first suggested applications of quantum computing was to simulate
quantum mechanical systems \cite{feynman1982simulating}. Since Feynman's
initial proposal, the field of quantum simulation has grown rapidly, resulting
in numerous algorithms that offer asymptotic speed-ups over the best possible
classical algorithms. Initially, much  focus has been on simulating closed quantum systems
\cite{berry2007efficient, berry2013gateefficient, berry2014exponential,
childs2019fasterquantum, Low2019hamiltonian, berry2015simulating,
campbell2019random,nakaji2024high}, only 
recently has there been a significant interest in developing quantum algorithms for simulating open quantum systems
\cite{childs2017efficient, dibartolomeo2023efficient, kliesch2011dissipative, li2023simulating,
li2023succinct, schlimgen2021quantum, schlimgen2022quantumsimulation,
schlimgen2022quantumstate, suri2023twounitary, hu2020a,
cleve2019efficient,guimares2023noise,guimares2024optimized,joo2023commutation}.
Notable examples are algorithms to simulate a single quantum channel applied to an input
quantum state \cite{hu2020a, dibartolomeo2023efficient, schlimgen2021quantum}, algorithms that simulate continuous-time dynamics describable by master equations~\cite{childs2017efficient, cleve2019efficient, dibartolomeo2023efficient,
kliesch2011dissipative, li2023simulating, li2023succinct,
schlimgen2022quantumsimulation}.  
Special
attention is paid to dynamics generated by time-independent Lindblad
master equations, since
they cover the common scenario of a quantum system weakly coupled to
a Markovian environment. Some examples include 
\cite{childs2017efficient, cleve2019efficient, dibartolomeo2023efficient,
kliesch2011dissipative, li2023simulating, li2023succinct,
schlimgen2022quantumsimulation,guimares2023noise,guimares2024optimized}
with \cite{li2023simulating} achieving the current state-of-the-art
asymptotic scaling. Similar to Hamiltonian simulation, the current 
classical Lindblad simulation algorithms are inefficient, thus efficient
quantum algorithms for Lindblad simulation could prove useful for future
scientific needs. 

Interestingly, many of the current Lindblad simulation algorithms share common
features.  They deterministically implement the full set of jump operators for
the Lindbladian and treat both the dissipative and unitary dynamics on equal
footing. In addition, the resources required by these algorithms explicitly
increase
(polynomially) with respect to the number of Lindblad jump operators along with
the overall norm of the Lindbladian. A missing feature in the existing
algorithms is taking advantage of the fact that in many realistic open system
dynamics, there are classical 
uncertainty regarding the effect of the environment on the system. 
This noise-agnostic approach can lead to situations where the noisy open system
simulation algorithms
can scale worse, sometimes even exponentially worse, than simulating the noiseless
system! For example, naive algorithms to simulate Hamiltonian evolution in the
presence of \textit{global} depolarizing noise would require exponential time
for many of the general-purpose open quantum system algorithms, given the fact
that global depolarization requires exponentially many jump operators. 
However, such an explicit scaling with the number of jump operators can be
 reduced if the classical randomness, present in the description
of the open system dynamics, is directly incorporated into the design of the
 algorithm to simulate dissipation. This strategy has recently been considered for a restricted class of Lindbladians, in which a first-order approximation to a dissipator is used for thermal state preparation on a quantum computer \cite{chifang2023quantum}. It is also used to simulate continuous-time dynamics by leveraging the intrinsic noise of quantum processors \cite{guimares2023noise,guimares2024optimized}, and to design an algorithm that uses a classical
processor alongside the quantum computer to reconstruct the dynamics based on a series of sampled states \cite{liu2024simulation}.

In this Letter, we present a quantum algorithm for simulating both the coherent and dissipative
evolution of a Lindbladian, with no need for any ancilla qubits and reduced
gate complexity in the number of jump operators, although for a restricted
class of Lindbladians. The class of 
Lindbladians we consider are sums
of a Hamiltonian part and a dissipator, where the dissipator is constrained to 
generate channels that are convex combinations of unitary channels and fixed
state preparation. 
Although our noise model is restrictive, it does
encapsulate many physically relevant dissipative processes such as both local and global depolarizing
noise, dephasing and bit-flip noise, and other random Pauli
channels.
The algorithm we introduce uses a second-order product formula to divide the generator into its Hamiltonian part and dissipator.
We then
simulate the dissipative portion using a simple sampling protocol, with the
Hamiltonian dynamics being simulated using any Hamiltonian simulation
algorithm preferred. The only ancillary qubits needed by the algorithm are what is needed to
implement the desired Hamiltonian simulation subroutine. The scaling of the
circuit with respect to simulation time and precision is
$O(T^{1.5}/\sqrt{\epsilon})$, where $T$ is the total simulation time and
$\epsilon$ is the accuracy in the diamond norm. This is an improvement over the
$O(T^{2}/\epsilon)$ scaling of
\cite{chifang2023quantum}.

We also discuss a generalization of the algorithm to simulate a family of
time-dependent Lindbladians. To analyze the performance in the time-dependent case, we prove a new
second-order product formula for time-dependent Louvillians, consisting of an arbitrary time-dependent and an arbitrary time-independent term, which might
be of independent interest.
Additionally, we explore how limited ancillary qubits can broaden the scope of
master equations that our algorithm can simulate, for example simulating local amplitude damping noise using a few ancilla qubits. We also discuss methods to extend the algorithm to non-Markovian master equations.

\textit{Stochastically Simulatable Channels.}---%
Let $\mathcal{E}$ be a CPTP map that can be simulated
stochastically, i.e., it can be implemented by sampling a  set of quantum operations $\{\mathcal{A}_{i}\}$ according to 
probabilities $p_{i}$ such that $
\mathcal{E}(\rho)=\mathbb{E}\left[\mathcal{A}_{i}(\rho)\right]$ holds. Then the measurement statistics of any quantum circuit composed of multiple applications of $\mathcal{E}$  can be reproduced by  replacing any application of
$\mathcal{E}$ with an implementation of $\mathcal{A}_{i}$
sampled independently. In this work we are interested in estimating the expected value of a general observable using our simulation algorithm. With
the goal of designing a low-cost quantum algorithm,  we restrict the operations $\mathcal{A}_{i}$ to be either
unitary or fixed state preparation operations. More
specifically, we consider stochastically simulatable  
channels defined below.
\begin{definition}
    \label{def:stocsimlindblad}
    We call $\mathcal{E}$ a stochastically simulatable channel if it can be written as:
    \begin{align}
    \mathcal{E}(\rho) =
    q\left(\sum_{i}\lambda_{i}U_{i}\rho U_{i}^{\dagger}\right) + 
    (1 - q)\rho_{f}
    \end{align}
    where $U_{i}$ are unitary operators and $\rho_{f}$ is a fixed quantum
    state. Scalars $q$ and
    $\lambda_{i}$ are constrained such that $0 \leq q \leq 1$,  $0 \leq
    \lambda_{i} \leq 1$, and $\sum_{i}\lambda_{i} = 1$.
\end{definition}
The cost of implementation of $\mathcal{E}$  solely depends on the
complexity of implementation of $\{U_{i}\}$  and also the preparation of $\rho_{f}$. Ideally these operations should be efficiently implementable.
Some examples include when $\{U_{i}\}$ are $n$-fold tensor products of
single-qubit unitaries and also when they are generated by
Clifford circuits. 

In this work we mainly focus on simulating open quantum systems described by a
Lindblad master equation, given as $
\frac{d\rho}{dt} = \mathcal{L}(\rho)$ with $\mathcal{L}= \mathcal{H}+\mathcal{D}$, where $\mathcal{H}(\rho) = -i [H, \rho]$ generates Hamiltonian evolution and  the dissipator $\mathcal{D}(\rho) = \sum_{\mu=1}^{m}\left(L_{\mu}\rho L_{\mu}^{\dagger} -
\frac{1}{2}\{L_{\mu}^{\dagger}L_{\mu}, \rho\}\right)$ models the
system-environment interactions. Here $H$ is the system Hamiltonian and $L_{\mu}$ are the Lindblad's jump
operators.

An example of a class of Lindblad superoperators  that satisfy
Definition~\ref{def:stocsimlindblad}
for a small interval of time are dissipators $\mathcal{D}$ with jump
operators $L_{\mu} = \alpha_{\mu}U_{\mu}$ for unitaries $U_{\mu}$. It is straightforward to check that in this case when $\delta t \dnorm{\mathcal{D}}\leq 1$, the operator $e^{\delta t \mathcal{D}}$ can be approximated by a random unitary channel 
with arbitrarily small error [see Supplementary Material (SM)]. When $U_{\mu}$ belongs to a group of
efficiently implementable unitaries, such as $n$-fold tensor products of single-qubit unitaries or Clifford circuits, then the random unitary channel fits our ideal scenario discussed above, since the complexity of
implementation does not change when the unitaries are composed. This class of
Lindblad master equations captures many useful families including random Pauli channels.

\textit{New Algorithm.}---%
We now introduce a quantum algorithm to simulate Lindbladian dynamics. 
In its simplest form, given a
time-independent Lindblad superoperator  $\mathcal{L} =
\mathcal{H} + \mathcal{D}$ where $\mathcal{H}$ is the generator of the unitary evolution and  $\mathcal{D}$ is the generator of a (approximately)  stochastically
simulatable map, our proposed algorithm simulates the evolution of an arbitrary
input state $\rho$ under $\mathcal{L}$ 
up to time $T$ and precision $\epsilon$ by repeatedly implementing Lindbladian simulation gadgets depicted in FIG~\ref{fig:alg1gadget}.
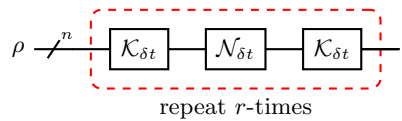
\begin{figure}[h]
    \centering
    \begin{quantikz}
        \lstick{$\rho$} & \qwbundle{n} & \gate{\mathcal{K}_{\delta
        t}}\gategroup[1, steps=3, style={dashed,rounded corners,
        red},background,label style={label
        position=below,anchor=north,yshift=-0.2cm}]{repeat $r$-times} & 
        \gate{\mathcal{N}_{\delta t}}& \gate{\mathcal{K}_{\delta t}} & \qw
    \end{quantikz}
    \caption{Simulation gadget}
    \vspace*{-5pt} 
    \label{fig:alg1gadget}
\end{figure}

Each Lindbladian simulation gadget is built of two different subroutines, one implementing the map $\mathcal{K}_{\delta t}$ which approximates $e^{\frac{\delta t}{2}\mathcal{H}}$ and can be executed using
any choice of Hamiltonian simulation algorithm,
and the other, $\mathcal{N}_{\delta t}$ which approximates $e^{\delta t
\mathcal{D}}$ in expectation using a simple sampling routine. 

As discussed before, we assume  $e^{\delta t \mathcal{D}}$ is approximately stochastically simulatable, i.e., 
   $e^{\delta t \mathcal{D}} = \mathcal{N}_{\delta t} +
    \mathcal{O}(\delta t^{3})$, where  $\mathcal{N}_{\delta t}$ satisfies
Definition~\ref{def:stocsimlindblad}. Since $\mathcal{N}_{\delta t}$ is stochastically simulatable we can write
\begin{equation}
    \mathcal{N}_{\delta t}(\rho) = \sum_{i}p_{i}(\delta
    t)\mathcal{A}_{\delta t}^{(i)}(\rho) =
    \mathbb{E}_{i \sim p_{i}(\delta t)}\left[\mathcal{A}_{\delta t}^{(i)}(\rho)\right],
\end{equation}
where:
\begin{equation} \label{eq:operation}
    \mathcal{A}_{\delta t}^{(i)}(\rho) = 
    \begin{cases} P_{\delta t}(\rho) = \rho_{\delta t} \text{ for } i = 0,\\
        U_{\delta t}^{(i)}\rho
        U_{\delta t}^{(i) \dagger}
    \text{ otherwise },
    \end{cases}
\end{equation}
along with probabilities 
\begin{equation} \label{eq:prob}
    p_{i}(\delta t) = \begin{cases}1 - q(\delta t) \text{ for } i = 0 ,
     \\ q(\delta t)\lambda_{i}(\delta t)
    \text{ otherwise}. \end{cases}
\end{equation}
Therefore our subroutine $\mathcal{N}_{\delta t}$
can approximate $e^{\delta t \mathcal{D}}$ up to
$\mathcal{O}(\delta t^{3})$, by simply sampling $i \sim
p_{i}(\delta t)$ and applying the
operation $\mathcal{A}_{\delta t}^{(i)}$. Thus the gate complexity of
$\mathcal{N}_{\delta t}$ is independent of
the number of jump operators and only
depends on the complexity of implementing unitaries $U_{\delta t}^{(i)}$ and
fixed state preparation oracle $P_{\delta t}$ which 
replaces the state of the system with a fixed state.  
This procedure assures
the measurement statistics over many different compilations of the quantum circuit
implementing $\mathcal{N}_{\delta t}$ will replicate the measurement statistics
of a single circuit implementing $\mathcal{N}_{\delta t}$ exactly. Thus when
compiling each of the Lindbladian
simulation gadgets, we independently sample $i \sim p_{i}(\delta t)$ and implement the operation
$\mathcal{A}_{\delta t}^{(i)}$ in place of $\mathcal{N}_{\delta t}$ in the
gadget. The full algorithm consists of sequentially repeating the blocks of
$\mathcal{K}_{\delta t}$ using any closed-system quantum simulation algorithm
of choice, and sampling $\mathcal{A}_{\delta t}^{(i)}$ for each $\mathcal{N}_{\delta t}$ block. The number of repetitions $r$, required to guarantee a desired accuracy $\epsilon$, would determine the gate complexity of the algorithm.

\textit{Algorithm Analysis.}---%
At a high level, to analyze the performance of the algorithm, we leverage a product formula for Lindblad
superoperators to split up the
dynamics into its Hamiltonian part and a stochastically simulatable 
dissipator. The total simulation time and precision then constrain the
individual subroutines' precision and simulation time. Since we can implement
$\mathcal{N}_{\delta t}$ without error, our constraints only apply to the
Hamiltonian simulation subroutine $\mathcal{K}_{\delta t}$.

First, we introduce some notation. The diamond norm of a linear map
$\mathcal{L}$ over linear operators on a finite-dimensional Hilbert space is
defined as $\dnorm{\mathcal{L}} = \max_{\norm{\rho}_{1} =
1}\norm{\mathcal{I}\otimes\mathcal{L}(\rho)}_{1},$ where $\norm{\cdot}_{1}$ denotes the trace norm  and
$\mathcal{I}$  the identity map. Following  \cite{cleve2019efficient} we  define 
$\pnorm{\mathcal{L}}$ of a Lindblad superoperator $\mathcal{L}$ as 
$\norm{\mathcal{L}}_{pauli} = \sum_{k=0}^{s-1}\beta_{0k} +\sum_{\mu=1}^{m}\left(\sum_{k=0}^{s-1}\beta_{\mu k}\right)^{2}$,
where $H = \sum_{k=0}^{s-1}\beta_{0k}V_{0k}$, and  $L_{\mu} = \sum_{k=0}^{s-1}\beta_{\mu k}V_{\mu k}$ with $\beta_{\mu k} > 0$ and $V_{\mu k}$ an $n$-fold tensor product of Pauli
operators with arbitrary phases. Finally, we denote the total simulation
time and precision by $T$ and $\epsilon$ respectively.

We proceed by finding an upper-bound on 
the error between $r$ applications of our subroutines and the true evolution
generated by $\mathcal{L}$. Assuming     $\left(\frac{\dnorm{\mathcal{H}}}{2} + \dnorm{\mathcal{D}}\right)\delta t
\leq 1$, we can apply a second-order product formula \cite{werner2016positive} (see Theorem~\ref{thm:timedeptrott2} for our generalized version), alongside the sub-multiplicative and additive properties of the
    diamond norm \cite{watrous_2018} to get:
\begin{multline}
    \label{eq:alg1E}
    \norm{e^{T\mathcal{L}} - \left(\mathcal{K}_{\delta t} 
    \circ \mathcal{N}_{\delta t} \circ
    \mathcal{K}_{\delta t}\right)^{r}}_{\diamond} \leq \\ 
    \frac{\dnorm{\left[\mathcal{H}, \mathcal{D}\right]}}{3}
    \left(\frac{\dnorm{\mathcal{H}}}{2} + \dnorm{\mathcal{D}}\right)r\delta t^{3} + 
    2r\left(\epsilon_{H} + \epsilon_{D}\right),
\end{multline}
where $\epsilon_{H} = \dnorm{e^{\delta t \mathcal{H}} - \mathcal{K}_{\delta
t}}$ and $\epsilon_{D} = \dnorm{e^{\delta t\mathcal{D}} -
\mathcal{N}_{\delta t}}$, which represent the errors each of our subroutines
incur.

To formalize our assumption about the generator $\mathcal{D}$, we require
\begin{equation}
    \label{eq:epsilonD}
    \epsilon_{D} = \dnorm{ e^{\delta t \mathcal{D}}-\mathcal{N}_{\delta t} }
    \leq c_{0}(\dnorm{\mathcal{D}}\delta t)^{3},
\end{equation}
where $c_{0}$ is a constant determined by the specifics of the dissipator
$\mathcal{D}$.  As discussed before, dissipators with jump operators of the form $L_{\mu} =
\alpha_{\mu}U_{\mu}$ can be approximated by a stochastically simulatable
channel up to arbitrary small error. Such dissipators are an example of when
requirement~\eqref{eq:epsilonD} is satisfied. Using $\dnorm{\mathcal{L}} \leq 2
\pnorm{\mathcal{L}}$ from \cite{cleve2019efficient} to carry out the
upper bounds in terms of $\pnorm{\cdot}$ gives us a way to easily compare to
the state-of-the-art Lindblad simulation algorithms
\cite{cleve2019efficient,li2023succinct, liu2024simulation}. Thus our bound becomes,
\begin{multline}
    \label{eq:alg1Efin}
    \norm{e^{T\mathcal{L}} - \left(\mathcal{K}_{\delta t} 
    \circ \mathcal{N}_{\delta t} \circ
    \mathcal{K}_{\delta t}\right)^{r}}_{\diamond} \leq \\ 
    \frac{8}{3}
    r(1 + 6c_{0})\left(\pnorm{\mathcal{L}}\delta t\right)^{3} + 
    2r\epsilon_{H}.
\end{multline}

To keep the precision of the total
simulation over time $T$ with-in total error of
$O(\epsilon)$  equation \eqref{eq:alg1Efin} implies we must take $r = O\left(\sqrt{\frac{1 + 6c_{0}}{\epsilon}}(\pnorm{\mathcal{L}}T)^{1.5}\right),$ and $\epsilon_{H} = O(\epsilon
/ r)$. All that is left is to design an implementation of the Hamiltonian simulation
subroutine.

Choosing to use the truncated Taylor series algorithm of \cite{berry2015simulating} 
provides the gate complexity for the subroutine to be 
 $\tilde{O}\left(\delta t \norm{\mathcal{H}}_{pauli} s
n \log(\frac{\delta
t\norm{\mathcal{H}}_{pauli}}{\epsilon_{H}})
\right)$, where $n$ is the number of qubits and $s$ in the number of Pauli terms in the Hamiltonian.   
The total gate complexity is
just the subroutine's cost multiplied by $r$.
Since subroutine
$\mathcal{N}_{\delta t}$ does not require any ancilla to implement, the only
ancilla cost of the algorithm comes from our implementation of
$\mathcal{K}_{\delta t}$. Assuming we can reset the ancilla registers before
each $\mathcal{K}_{\delta t}$ subroutine,
algorithm of Ref.~\cite{berry2015simulating} gives us the ancilla cost for implementing
$\mathcal{K}_{\delta t}$ to be  
$
\tilde{O}\left(\log(s)\log(\delta t / \epsilon_{H})\right).$

Alternatively, we could use an algorithm based on product formulas 
\cite{lloyd1996universal, childs2021theory} to implement  the
Hamiltonian simulation subroutine $\mathcal{K}_{\delta t}$,
which would remove the need for any ancilla, at the expense of a worse gate
complexity to simulate the Hamiltonian dynamics. We can summarize the analysis in the following theorem.
\begin{theorem}
    \label{thm:costs}
    Let $\mathcal{L}=\mathcal{H}+\mathcal{D}$ be a Lindbladian  with Hamiltonian portion $\mathcal{H}$
    and dissipative portion $\mathcal{D}$. Assume $\mathcal{D}$ generates an
    approximately stochastically simulatable channel satisfying
    \eqref{eq:epsilonD}. Then there exists a quantum algorithm that simulates
    $e^{T \mathcal{L}}$ up to precision $\epsilon$ 
    that requires 
    \begin{equation}
        r = O\left(\sqrt{\frac{1 + 6c_{0}}{\epsilon}}(\pnorm{\mathcal{L}}T)^{1.5}\right)
    \end{equation}
    calls to oracles $\mathcal{A}_{\delta t}^{(i)}$, and
    any Hamiltonian simulation subroutine $\mathcal{K}_{\delta t}$ where
    $\delta t = T / r$. No ancilla
    qubits are needed beyond what is required for the Hamiltonian simulation
    subroutine $\mathcal{K}_{\delta t}$.
\end{theorem}

As an example, consider the case of a Lindbladian $\mathcal{L}_{\gamma}$
that is a sum  of an arbitrary
Hamiltonian part  $\mathcal{H}$ and a dissipator $\mathcal{D}_{\gamma}$ which generates the global
depolarization channel with rate $\gamma$.  Applying
Theorem~\ref{thm:costs} gives us the total cost to simulate the dynamics
described by $\mathcal{L}_{\gamma}$ to be $O\left(n  \frac{((\gamma +
\pnorm{\mathcal{H}})T)^{1.5}}{\sqrt{\epsilon}}\right)$
 gates along with 
$\Tilde{O}\left(\log(T / \epsilon)\right)$
ancilla if choosing to use the Hamiltonian simulation subroutine of
\cite{berry2015simulating}.
As another example, consider the case when $\mathcal{L}_{\Gamma}$
generates 1-local dephasing noise. In this case,
the jump
operators have the form $L_{\mu} = \sqrt{\frac{\Gamma}{2}}Z_{\mu}$ where
$Z_{\mu}$ acts on the $\mu$-th qubit. The analysis follows similar to the above
with  $\gamma$ replaced by $n\Gamma / 2$. For details of both calculations see the
SM.

\textit{Time-dependent Lindblad equation.}---%
The algorithm and analysis above can be generalized to include time-dependent
Lindblad superoperators 
with a time-dependent dissipative portion $\mathcal{D}(t)$. 
If the dissipator generates a channel that is approximately stochastically simulatable at
all times, i.e.,
\begin{equation}
    \mathcal{T}e^{\int_{t}^{t+\delta t}\mathcal{D}(t')dt'}(\rho(t)) =
    \mathcal{N}_{t,\delta t}(\rho(t)) + \mathcal{O}(\delta t^{3}),
\end{equation}
holds for all $t$, then the algorithm can be easily modified to simulate such
systems. To do so we leverage Theorem~\ref{thm:timedeptrott2} to divide the
time-dependent dynamics into its time-dependent dissipator and its
Hamiltonian part.
\begin{theorem}
    \label{thm:timedeptrott2}
    Let $\mathcal{H}$ and $\mathcal{D}(t)$ be Lindbladians and $0 \leq
    \left(\frac{\dnorm{\mathcal{H}}}{2} + \sup_{t}\dnorm{\mathcal{D}(t)}\right)\delta t
    \leq 1$, then
    \begin{align}
&\dnorm{\mathcal{T}e^{\int_{t}^{t+\delta t}(\mathcal{H} +\mathcal{D}(t'))dt'} - e^{\frac{\delta t}{2}\mathcal{H}}\circ\mathcal{T}e^{\int_{t}^{t+\delta t}\mathcal{D}(t')dt'} \circ e^{\frac{\delta t}{2}\mathcal{H}}}  \nonumber \\
& \quad  \leq \frac{\sup_{t}\dnorm{[\mathcal{H}, \mathcal{D}(t)]}}{3}\left(\frac{\dnorm{\mathcal{H}}}{2}  \sup_{t}\dnorm{\mathcal{D}(t)}\right)\delta t^{3}.
    \end{align} 
\end{theorem}
\begin{proof}
    See Supplementary Material.
\end{proof}
Similar to the time-independent case, we implement subroutines to approximate
the dynamics generated by these two generators. The only difference lies in how
we implement the dissipative subroutine. In the time-dependent case, we
implement the stochastically simulatable channels $\mathcal{N}_{t,\delta t}$ as
our subroutine approximating the dissipative dynamics. Since
$\mathcal{N}_{t,\delta t}$ explicitly depends on time, we will implement
different stochastically simulatable channels at each time step $t$ where we
require a dissipative subroutine. The rest of the analysis of the algorithm is similar to the time-independent
case.  The total costs as shown above hold in
the time-dependent scenario with the only difference being, replacing
$\pnorm{\mathcal{L}}$ with $ \sup_{t}\pnorm{\mathcal{L}(t)}$ and $c_{0}$ with
$\sup_{t}c_{0}(t)$ in Theorem~\ref{thm:costs}.

\textit{Generalized noise-model by allowing ancilla qubits.}---%
At the expense of
adding ancilla qubits and applying unitary operations on the joint
system, one can extend the scope of the proposed method to
simulate a larger family of noise models. The Stinespring dilation theorem
guarantees that any CPTP map acting on $2^l\times 2^l$ density operators can be
implemented by adding  $2l$ ancilla qubits and performing a (generally
complicated) unitary operation on the joint system \cite{Nielsen_Chuang_2010}.
However, the probabilistic implementation of channels can reduce this ancilla
requirement.
\begin{remark}
The algorithm can implement any  CPTP map acting on  $2^l\times 2^l$ density operators using at most $l$ ancilla qubits.
\end{remark}

This observation follows from the fact that any such a map can be expressed as 
a convex combination of  
extremal channels with each
extremal channel only having at most $2^l$ many Kraus operators
\cite{choi1975completely}. Therefore implementation of dilatation of any of the extremal channels requires at most $l$ ancilla qubits. Given that only one extremal channel needs to be implemented in each compilation, at most $l$ ancilla suffices in each run. The ancilla qubits can be reset after the application of the channel.

Simulation of local noise, naturally arising in many physical systems, is of particular interest. A local (but possibly correlated) noise model can be described by a CPTP map,
\begin{equation}
    \mathcal{E}(\rho) =
    \sum_{i}\lambda_{i}\mathcal{G}^{(i)}_{0}\otimes ... \otimes
    \mathcal{G}^{(i)}_{n}(\rho),
\end{equation}
where each $\mathcal{G}^{(i)}_{j}$ acts on the $j$th qubit, with $l=1$. If each
$\mathcal{G}^{(i)}_{j}$ is unital then we can expand each operation
in terms of a convex combination of single-qubit unitaries. This gives us a
description of the operation that satisfies Definition~\ref{def:stocsimlindblad}
where the $U_{i}(t)$ operations are $n$-fold tensor products of
single-qubit unitaries, thus efficiently implementable. If 
$\mathcal{G}^{(i)}_{j}$ are not unital, for example such as amplitude damping, then we only need one ancilla rather than two, to implement each of the local maps. Therefore the full noisy process $\mathcal{E}(\rho)$ can be simulated by adding at most $n$ extra
ancilla qubits and 2-qubit gates, which is less than the 
$2n$ ancilla
needed if we chose to directly dilate.

\textit{Incorporating time-correlations.}---%
As discussed before, to simulate the Markovian master equation, we implement each dissipation block of $\mathcal{N}_{t,\delta t}$ by generating samples independent from the other blocks. However, it is straightforward to modify the algorithm to simulate dynamics that have correlations in time by simply introducing correlations in the sampling of dissipation blocks used in the compilation. In particular, the probabilities and operations in equations \eqref{eq:operation} and \eqref{eq:prob} can be generalized to define joint probabilities and operations over multiple time steps. Therefore, the domain of master equations simulable by the presented algorithm goes beyond the Markovian master equation we focused on in this work.

\textit{Discussions.}---%
We have presented a novel quantum algorithm for simulating
restricted, but physically motivated, types of dynamics generated by the Lindblad master equations. Our algorithm leverages classical
randomness in the compilation phase of the quantum circuit to circumvent the
overhead of deterministically implementing all of the jump operators needed to simulate the dynamics. Our approach eliminates the
explicit scaling on the number
of jump operators as it appears in previous algorithms, with the only scaling depending on
the norm of the Lindblad superoperator, which all other Lindblad simulation circuit
scalings depend on. This strategy also has the benefit of reducing ancilla costs since the only
ancilla requirements come from the choice of Hamiltonian
simulation subroutine. We also
treat the closed system dynamics on a different footing than the dissipative
dynamics. This approach allows us to leverage the wealth of available closed-system quantum algorithms in addition to extending the algorithm to deal with scenarios where dissipation may be time-dependent or correlated over time.

It would
be interesting to see if our approach of leveraging classical randomness and
treating the dissipative evolution differently than the closed system evolution
can be used to achieve better scaling with respect to $T$ and
$\epsilon$, ideally $O(T \log(\frac{1}{\epsilon}))$ of \cite{cleve2019efficient} but without the dependence on the number of jump operators and using only a few ancillary qubits. Extending the
algorithm to incorporate more general non-Markovian dynamics is another interesting future direction.

\begin{acknowledgments}
\textit{Acknowledgments.}---%
This work is supported by Sandia National Laboratories’ Laboratory Directed Research and Development program (Contract \#2534192). Additional support by DOE's Express: 2023 Exploratory Research For Extreme-scale Science Program under Award Number DE-SC0024685 is acknowledged.
\end{acknowledgments}

\bibliography{refs}
\bibliographystyle{apsrev4-2}

\clearpage
\setcounter{table}{0}
\renewcommand{\thetable}{S\arabic{table}}%
\setcounter{figure}{0}
\renewcommand{\thefigure}{S\arabic{figure}}%
\setcounter{section}{0}
\setcounter{equation}{0}
\renewcommand{\theequation}{S\arabic{equation}}%

\onecolumngrid

\section{Time-Dependent second Order Product formula}  
\label{Sectimedep}

In this section, we prove a time-dependent second-order product formula bound for
Liouvillians. For simplicity, we assume that the Liouvillian $\mathcal{L}$ splits
up into a time-independent component $\mathcal{H}$ and a time-dependent
component $\mathcal{D}(t)$. This bound is used in the main text when
generalizing our algorithm to incorporate time-dependent dissipators.

Before proving the main result we cite an important Lemma from
Kliesch et al. \cite{kliesch2011dissipative}  that will be useful in the proof.
\begin{lemma}[Backward time evolution \cite{kliesch2011dissipative}]
    \label{lm:bte}
    Let $\mathcal{L}$ be a Liouvillian. For $t > s$:
    \begin{enumerate}
        \item $T_{\mathcal{L}}(t,s)=\mathcal{T}e^{\int_{s}^{t}\mathcal{L}(t')dt'}$ is invertible and its inverse is
            $T_{\mathcal{L}}^{-}(t,s)$ as defined by \begin{equation}
            \partial_{t}T_{\mathcal{L}}^{-}(t,s) = -
        T_{\mathcal{L}}^{-}(t,s)\mathcal{L} \end{equation}
        \item If the Liouvillian $\mathcal{L}$ is piecewise continuous in time
            then \begin{equation} \dnorm{T_{\mathcal{L}}^{-}(t,s)} \leq
            e^{\int_{s}^{t}\dnorm{\mathcal{L}(t')}dt'} \end{equation}
    \end{enumerate}

\end{lemma}
\begin{proof}
    See \cite{kliesch2011dissipative}.
\end{proof}
We now prove the main result. 
\begin{theorem*}[3] (restated.) 
    Let $\mathcal{H}$ and $\mathcal{D}(t)$ be Liouvillians and $0 \leq
    \left(\frac{\dnorm{\mathcal{H}}}{2} +
    \sup_{t}\dnorm{\mathcal{D}(t)}\right)(t-s)
    \leq 1$, then
    \begin{equation}
        \label{eq:timedepbound}
        \dnorm{\mathcal{T}e^{\int_{s}^{t}(\mathcal{H} +
        \mathcal{D}(t'))dt'} - e^{\frac{(t-s)}{2}
        \mathcal{H}}\circ\mathcal{T}e^{\int_{s}^{t}\mathcal{D}(t')dt'} \circ
        e^{\frac{(t-s)}{2}\mathcal{H}}} \leq
        \frac{\sup_{t}\dnorm{[\mathcal{H},
        \mathcal{D}(t)]}}{3}\left(\frac{\dnorm{\mathcal{H}}}{2} +
        \sup_{t}\dnorm{\mathcal{D}(t)}\right)(t-s)^{3}. 
    \end{equation}
\end{theorem*}
\begin{proof}
    Define $T_{\mathcal{L}}(t, s) =
    \mathcal{T}e^{\int_{s}^{t}\mathcal{L}(t')dt'}$ and $T_{\mathcal{L}}^{-}(t,
    s)$ according to Lemma~\ref{lm:bte}. Both
    $T_{\mathcal{L}}(t, s)$ and $T_{\mathcal{L}}^{-}(t,s)$ are the unique
    solutions to the differential equations 
    \begin{equation}
        \partial_{t}T_{\mathcal{L}}(t,s)
        = \mathcal{L}(t) T_{\mathcal{L}}(t,s) \text{ and }
        \partial_{t}T_{\mathcal{L}}^{-}(t,s) = -T_{\mathcal{L}}^{-}(t,s)
        \mathcal{L}(t)
    \end{equation}
    with the initial condition $T_{\mathcal{L}}(s,s) = T_{\mathcal{L}}^{-}(s,s) =
    \mathcal{I}$.
    Rewriting the left hand side of equation~\eqref{eq:timedepbound}
    in terms of our definitions gives
    \begin{equation}
        \epsilon = \dnorm{T_{\mathcal{H} + \mathcal{D}}(t,s) - T_{\mathcal{H}}\left(t,
        \frac{t+s}{2}\right)T_{\mathcal{D}}(t,
        s)T_{\mathcal{H}}\left(\frac{t+s}{2}, s\right)},
    \end{equation}
    which can be rewritten as
    \begin{align}
        \label{eq:redefs}
        \epsilon &=
        \dnorm{T_{\mathcal{H}}\left(t,
        \frac{t+s}{2}\right)T_{\mathcal{D}}(t,
        s)T_{\mathcal{H}}\left(\frac{t+s}{2},
        s\right)\left(T_{\mathcal{H}}^{-}\left(\frac{t+s}{2},
        s\right)T_{\mathcal{D}}^{-}(t,s)T_{\mathcal{H}}^{-}\left(t,\frac{t+s}{2}\right)T_{\mathcal{H}
        + \mathcal{D}}(t,s) -
        \mathcal{I} \right)} \\ 
        &= \dnorm{T_{\mathcal{H}}\left(t,
        \frac{t+s}{2}\right)T_{\mathcal{D}}(t,
        s)T_{\mathcal{H}}\left(\frac{t+s}{2},
        s\right)F(t,s)}
    \end{align}
    where we have defined $F(t,s) = T_{\mathcal{H}}^{-}\left(\frac{t+s}{2},
    s\right) T_{\mathcal{D}}^{-}(t,s)
    T_{\mathcal{H}}^{-}\left(t,\frac{t+s}{2}\right)
    T_{\mathcal{H}+\mathcal{D}}(t,s) -
    \mathcal{I}$. Focusing on $F(t,s)$ and using the
    fundamental theorem of calculus one finds:
    \begin{align}
        F(t,s) &= \int_{s}^{t}\partial_{r}\left(T_{\mathcal{H}}^{-}\left(\frac{r+s}{2},
        s\right) T_{\mathcal{D}}^{-}(r,s)
        T_{\mathcal{H}}^{-}\left(r,\frac{r+s}{2}\right)
        T_{\mathcal{H}+\mathcal{D}}(r,s)\right)dr \nonumber \\
        &= \int_{s}^{t}\left\{\partial_{r}\left[T_{\mathcal{H}}^{-}\left(\frac{r+s}{2},
        s\right)\right] T_{\mathcal{D}}^{-}(r,s)
        T_{\mathcal{H}}^{-}\left(r, \frac{r+s}{2}\right)
        T_{\mathcal{H}+\mathcal{D}}(r,s) \right. \nonumber \\
        &+ T_{\mathcal{H}}^{-}\left(\frac{r+s}{2}, s\right)
            \partial_{r}\left[T_{\mathcal{D}}^{-}(r,s)\right]
        T_{\mathcal{H}}^{-}\left(r,\frac{r+s}{2}\right)
        T_{\mathcal{H}+\mathcal{D}}(r,s) \nonumber \\
        &+ T_{\mathcal{H}}^{-}\left(\frac{r+s}{2}, s\right)
        T_{\mathcal{D}}^{-}(r,s)
        \partial_{r}\left[T_{\mathcal{H}}^{-}\left(r,\frac{r+s}{2}\right)\right]
        T_{\mathcal{H}+\mathcal{D}}(r,s) \nonumber \\
        & \left. + T_{\mathcal{H}}^{-}\left(\frac{r+s}{2}, s\right)
        T_{\mathcal{D}}^{-}(r,s)
        T_{\mathcal{H}}^{-}\left(r,\frac{r+s}{2}\right)
        \partial_{r}\left[T_{\mathcal{H}+\mathcal{D}}(r,s)\right]\right\}dr.
    \end{align}
    Applying identities: $\partial_{t}T_{\mathcal{H}+\mathcal{D}}(t,s) =
    (\mathcal{H} + \mathcal{D}(t))T_{\mathcal{H}+\mathcal{D}}(t,s)$,
    $\partial_{t}T_{\mathcal{H}}^{-}\left(\frac{t + s}{2},
    s\right) = -\frac{1}{2} T_{\mathcal{H}}^{-}\left(\frac{t + s}{2},
    s\right)\mathcal{H}$, 
    $\partial_{t}T_{\mathcal{H}}^{-}\left(t, \frac{t + s}{2}\right) =
    -\frac{1}{2}T_{\mathcal{H}}^{-}\left(t, \frac{t + s}{2}\right)\mathcal{H}$,
    and $\partial_{t}T_{\mathcal{D}}^{-}(t,s) = -
    T_{\mathcal{D}}^{-}(t,s)\mathcal{D}(t)$ yields,
    \begin{align}
        F(t,s) &= \int_{s}^{t}T_{\mathcal{H}}^{-}\left(\frac{r+s}{2},
        s\right)\left\{-\frac{\mathcal{H}}{2} T_{\mathcal{D}}^{-}(r,s)
        T_{\mathcal{H}}^{-}\left(r,\frac{r+s}{2}\right) \right. \nonumber \\
        &- T_{\mathcal{D}}^{-}(r,s)\mathcal{D}(r)
        T_{\mathcal{H}}^{-}\left(\frac{r+s}{2},s\right) \nonumber \\
        &- 
        T_{\mathcal{D}}^{-}(r,s)
        T_{\mathcal{H}}^{-}\left(r,\frac{r+s}{2}\right)\frac{\mathcal{H}}{2} \nonumber \\
        & \left. + T_{\mathcal{D}}^{-}(r,s)
        T_{\mathcal{H}}^{-}\left(r,\frac{r+s}{2}\right)
        (\mathcal{H} +
        \mathcal{D}(r))\right\}T_{\mathcal{H}+\mathcal{D}}(r,s)dr \nonumber \\
        &= \int_{s}^{t}T_{\mathcal{H}}^{-}\left(\frac{r+s}{2},
        s\right)\left\{-\frac{\mathcal{H}}{2} T_{\mathcal{D}}^{-}(r,s)
        T_{\mathcal{H}}^{-}\left(r,\frac{r+s}{2}\right) \right. \nonumber \\
        &- T_{\mathcal{D}}^{-}(r,s)\mathcal{D}(r)
        T_{\mathcal{H}}^{-}\left(\frac{r+s}{2},s\right) \nonumber \\
        &\left. + T_{\mathcal{D}}^{-}(r,s)
        T_{\mathcal{H}}^{-}\left(r,\frac{r+s}{2}\right)
        \left(\frac{\mathcal{H}}{2} +
        \mathcal{D}(r)\right)\right\}T_{\mathcal{H}+\mathcal{D}}(r,s)dr\nonumber \\
        &= -\int_{s}^{t}T_{\mathcal{H}}^{-}\left(\frac{r+s}{2},
        s\right)\left\{\left[T_{\mathcal{D}}^{-}(r,s),\mathcal{D}(r)\right]
        T_{\mathcal{H}}^{-}\left(\frac{r+s}{2},s\right)\right. \nonumber \\
        &\left. + \left[\frac{\mathcal{H}}{2} +
                \mathcal{D}(r),T_{\mathcal{D}}^{-}(r,s)
        T_{\mathcal{H}}^{-}\left(r,\frac{r+s}{2}\right)\right]\right\}T_{\mathcal{H}+\mathcal{D}}(r,s)dr.
    \end{align} 
    Applying the fundamental theorem of calculus to the terms inside each
    of the commutators and using the above derivative identities again yields,
    \begin{align}
        F(t,s) &= -\int_{s}^{t}T_{\mathcal{H}}^{-}\left(\frac{r+s}{2},
            s\right)\int_{s}^{r}\left\{\partial_{u}\left(\left[T_{\mathcal{D}}^{-}(u,s),\mathcal{D}(r)\right]
        T_{\mathcal{H}}^{-}\left(\frac{u+s}{2},s\right)\right)\right. \nonumber \\
               &\left. +\partial_{u} \left(\left[\frac{\mathcal{H}}{2} +
        \mathcal{D}(r),T_{\mathcal{D}}^{-}(u,s)
        T_{\mathcal{H}}^{-}\left(u,\frac{u+s}{2}\right)\right]\right)\right\}
        T_{\mathcal{H}+\mathcal{D}}(r,s)du
        dr \nonumber \\
        &= -\int_{s}^{t}T_{\mathcal{H}}^{-}\left(\frac{r+s}{2},
            s\right)
            \int_{s}^{r}\left\{\left[\partial_{u}\left[T_{\mathcal{D}}^{-}(u,s)\right],\mathcal{D}(r)\right]
        T_{\mathcal{H}}^{-}\left(\frac{u+s}{2},s\right)\right. \nonumber \\
        &+ \left[T_{\mathcal{D}}^{-}(u,s),\mathcal{D}(r)\right]\partial_{u}\left[
        T_{\mathcal{H}}^{-}\left(\frac{u+s}{2},s\right)\right] \nonumber \\
        &\left. + \left[\frac{\mathcal{H}}{2} +
            \mathcal{D}(r), \partial_{u}\left[T_{\mathcal{D}}^{-}(u,s)\right]
        T_{\mathcal{H}}^{-}\left(u,\frac{u+s}{2}\right) + T_{\mathcal{D}}^{-}(u,s)\partial_{u}\left[
        T_{\mathcal{H}}^{-}\left(u,\frac{u+s}{2}\right)\right]\right] \right\}
        T_{\mathcal{H}+\mathcal{D}}(r,s) du dr \nonumber \\
        &= \int_{s}^{t}T_{\mathcal{H}}^{-}\left(\frac{r+s}{2},
            s\right)\int_{s}^{r}\left\{\left[T_{\mathcal{D}}^{-}(u,s)\mathcal{D}(u),\mathcal{D}(r)\right]
        T_{\mathcal{H}}^{-}\left(\frac{u+s}{2},s\right) +
        \left[T_{\mathcal{D}}^{-}(u,s),\mathcal{D}(r)\right]
        T_{\mathcal{H}}^{-}\left(\frac{u+s}{2},s\right)\frac{\mathcal{H}}{2}
        \right. \nonumber \\
        &\left. + \left[\frac{\mathcal{H}}{2} +
                \mathcal{D}(r), T_{\mathcal{D}}^{-}(u,s)\mathcal{D}(u)
        T_{\mathcal{H}}^{-}\left(u,\frac{u+s}{2}\right) + T_{\mathcal{D}}^{-}(u,s)
        T_{\mathcal{H}}^{-}\left(u,\frac{u+s}{2}\right)\frac{\mathcal{H}}{2}\right] \right\}
        T_{\mathcal{H}+\mathcal{D}}(r,s) du dr \nonumber \\
        &= \int_{s}^{t}\int_{s}^{r}T_{\mathcal{H}}^{-}\left(\frac{r+s}{2},
        s\right)\left\{T_{\mathcal{D}}^{-}(u,s)\left[\mathcal{D}(u),\mathcal{D}(r)\right]
        T_{\mathcal{H}}^{-}\left(\frac{u+s}{2},s\right) +
        \left[T_{\mathcal{D}}^{-}(u,s),\mathcal{D}(r)\right] \mathcal{D}(u)
        T_{\mathcal{H}}^{-}\left(\frac{u+s}{2},s\right) \right.\nonumber \\
        &+ \left[T_{\mathcal{D}}^{-}(u,s),\mathcal{D}(r)\right]\frac{\mathcal{H}}{2}
        T_{\mathcal{H}}^{-}\left(\frac{u+s}{2},s\right) \nonumber \\
        &\left. + \left[\frac{\mathcal{H}}{2} + \mathcal{D}(r),
        T_{\mathcal{D}}^{-}(u,s)\left(\frac{\mathcal{H}}{2} + \mathcal{D}(u)\right)
        T_{\mathcal{H}}^{-}\left(u,\frac{u+s}{2}\right)\right] \right\}
        T_{\mathcal{H}+\mathcal{D}}(r,s) du dr,
    \end{align} 
    where the last equality follows from commuting $T_{\mathcal{H}}^{-}$ and
    $\mathcal{H}$ since $\mathcal{H}$ is time-independent. Massaging the last
    commutator in the sum yields, 
    \begin{align}
        F(t,s) &= \int_{s}^{t}\int_{s}^{r}T_{\mathcal{H}}^{-}\left(\frac{r+s}{2},
        s\right)\left\{T_{\mathcal{D}}^{-}(u,s)\left[\mathcal{D}(u),\mathcal{D}(r)\right]
        T_{\mathcal{H}}^{-}\left(\frac{u+s}{2},s\right) +
        \left[T_{\mathcal{D}}^{-}(u,s),\mathcal{D}(r)\right] \mathcal{D}(u)
        T_{\mathcal{H}}^{-}\left(\frac{u+s}{2},s\right) \right.\nonumber \\
        &+ \left[T_{\mathcal{D}}^{-}(u,s),\mathcal{D}(r)\right]\frac{\mathcal{H}}{2}
        T_{\mathcal{H}}^{-}\left(\frac{u+s}{2},s\right) \nonumber \\
        &+ T_{\mathcal{D}}^{-}(u,s)\left(\frac{\mathcal{H}}{2} +
        \mathcal{D}(u)\right)\left[\mathcal{D}(r),
        T_{\mathcal{H}}^{-}\left(u,\frac{u+s}{2}\right)\right] \nonumber \\
        &+ T_{\mathcal{D}}^{-}(u,s)\left[\frac{\mathcal{H}}{2}, \mathcal{D}(u)
        \right]T_{\mathcal{H}}^{-}\left(u,\frac{u+s}{2}\right) +
        T_{\mathcal{D}}^{-}(u,s)\left[\mathcal{D}(r),
        \frac{\mathcal{H}}{2}\right]
        T_{\mathcal{H}}^{-}\left(u,\frac{u+s}{2}\right)\nonumber \\ 
        &+ T_{\mathcal{D}}^{-}(u,s)\left[\mathcal{D}(r),
        \mathcal{D}(u)\right]
        T_{\mathcal{H}}^{-}\left(u,\frac{u+s}{2}\right)\nonumber \\
        &+ \left[\frac{\mathcal{H}}{2},
        T_{\mathcal{D}}^{-}(u,s)\right] \left(\frac{\mathcal{H}}{2} +
        \mathcal{D}(u)\right) T_{\mathcal{H}}^{-}\left(u,\frac{u+s}{2}\right) \nonumber \\
        &\left. + \left[\mathcal{D}(r),
        T_{\mathcal{D}}^{-}(u,s)\right] \left(\frac{\mathcal{H}}{2} +
        \mathcal{D}(u)\right) T_{\mathcal{H}}^{-}\left(u,\frac{u+s}{2}\right)\right\}
        T_{\mathcal{H}+\mathcal{D}}(r,s) du dr.
    \end{align} 
    Simplifying the equation gives us, 
    \begin{align}
        F(t,s) &= \int_{s}^{t}\int_{s}^{r}T_{\mathcal{H}}^{-}\left(\frac{r+s}{2},
        s\right)\left\{- T_{\mathcal{D}}^{-}(u,s)\left(\frac{\mathcal{H}}{2} +
        \mathcal{D}(u)\right)\left[T_{\mathcal{H}}^{-}\left(u,\frac{u+s}{2}\right), 
        \mathcal{D}(r)\right] \right. \nonumber \\
        &- \left[T_{\mathcal{D}}^{-}(u,s),\frac{\mathcal{H}}{2}
        \right] \left(\frac{\mathcal{H}}{2} +
        \mathcal{D}(u)\right)
        T_{\mathcal{H}}^{-}\left(u,\frac{u+s}{2}\right)\nonumber \\
        &\left. + T_{\mathcal{D}}^{-}(u,s)\left[\frac{\mathcal{H}}{2}, \mathcal{D}(u)
        - \mathcal{D}(r)\right]
        T_{\mathcal{H}}^{-}\left(u,\frac{u+s}{2}\right)\right\}
        T_{\mathcal{H}+\mathcal{D}}(r,s) du dr\nonumber \\
        &= \int_{s}^{t}\int_{s}^{r}T_{\mathcal{H}}^{-}\left(\frac{r+s}{2},
        s\right)\Big\{\nonumber \\
        &- T_{\mathcal{D}}^{-}(u,s)\left(\frac{\mathcal{H}}{2} +
        \mathcal{D}(u)\right)
        \left(\int_{s}^{u}\partial_{v}\left[T_{\mathcal{H}}^{-}\left(v,\frac{v+s}{2}\right)
        \mathcal{D}(r) T_{\mathcal{H}}\left(v,\frac{v+s}{2}\right)\right]dv\right)
        T_{\mathcal{H}}^{-}\left(u,\frac{u+s}{2}\right) \nonumber \\
        &- \left(\int_{s}^{u}\partial_{v}\left[T_{\mathcal{D}}^{-}(v,s)\frac{\mathcal{H}}{2}
            T_{\mathcal{D}}(v,s)\right]dv\right)T_{\mathcal{D}}^{-}(u,s) \left(\frac{\mathcal{H}}{2} +
        \mathcal{D}(u)\right)
        T_{\mathcal{H}}^{-}\left(u,\frac{u+s}{2}\right)\nonumber \\
        &+ T_{\mathcal{D}}^{-}(u,s)\left[\frac{\mathcal{H}}{2},
            \int_{r}^{u}\partial_{v}\mathcal{D}(v)dv\right]
        T_{\mathcal{H}}^{-}\left(u,\frac{u+s}{2}\right)\nonumber \\
        &\Big\}T_{\mathcal{H}+\mathcal{D}}(r,s) du dr .
    \end{align}
    Applying the fundamental theorem of calculus again yields
    \begin{align}
        F(t,s) &= \int_{s}^{t}\int_{s}^{r}T_{\mathcal{H}}^{-}\left(\frac{r+s}{2},
        s\right)\Big\{\nonumber \\
        &- T_{\mathcal{D}}^{-}(u,s)\left(\frac{\mathcal{H}}{2} +
        \mathcal{D}(u)\right)
        \left(\int_{s}^{u}\partial_{v}\left[T_{\mathcal{H}}^{-}\left(v,\frac{v+s}{2}\right)
        \mathcal{D}(r) T_{\mathcal{H}}\left(v,\frac{v+s}{2}\right)\right]dv\right)
        T_{\mathcal{H}}^{-}\left(u,\frac{u+s}{2}\right) \nonumber \\
        &- \left(\int_{s}^{u}\partial_{v}\left[T_{\mathcal{D}}^{-}(v,s)\frac{\mathcal{H}}{2}
            T_{\mathcal{D}}(v,s)\right]dv\right)T_{\mathcal{D}}^{-}(u,s) \left(\frac{\mathcal{H}}{2} +
        \mathcal{D}(u)\right)
        T_{\mathcal{H}}^{-}\left(u,\frac{u+s}{2}\right)\nonumber \\
        &+ T_{\mathcal{D}}^{-}(u,s)\left[\frac{\mathcal{H}}{2},
            \int_{r}^{u}\partial_{v}\mathcal{D}(v)dv\right]
        T_{\mathcal{H}}^{-}\left(u,\frac{u+s}{2}\right)\nonumber \\
        &\Big\}T_{\mathcal{H}+\mathcal{D}}(r,s) du dr.
    \end{align}
    Continuing to simplify our expression for $F(t,s)$ gives us
    \begin{align}
        F(t,s) &= \int_{s}^{t}\int_{s}^{r}T_{\mathcal{H}}^{-}\left(\frac{r+s}{2},
        s\right)\Big\{\nonumber \\
        & \int_{s}^{u}T_{\mathcal{D}}^{-}(u,s)\left(\frac{\mathcal{H}}{2} +
        \mathcal{D}(u)\right)
        T_{\mathcal{H}}^{-}\left(v,\frac{v+s}{2}\right)\left[\frac{\mathcal{H}}{2},
        \mathcal{D}(r)\right] T_{\mathcal{H}}\left(v,\frac{v+s}{2}\right)
        T_{\mathcal{H}}^{-}\left(u,\frac{u+s}{2}\right)dv \nonumber \\
        &+ \int_{s}^{u}
        T_{\mathcal{D}}^{-}(v,s)\left[\mathcal{D}(v),\frac{\mathcal{H}}{2}\right]
        T_{\mathcal{D}}(v,s) T_{\mathcal{D}}^{-}(u,s) \left(\frac{\mathcal{H}}{2} +
        \mathcal{D}(u)\right)
        T_{\mathcal{H}}^{-}\left(u,\frac{u+s}{2}\right) dv\nonumber \\
        &+ \int_{r}^{u}T_{\mathcal{D}}^{-}(u,s)\left[\frac{\mathcal{H}}{2},
        \partial_{v}\mathcal{D}(v)\right]
        T_{\mathcal{H}}^{-}\left(u,\frac{u+s}{2}\right)dv\nonumber \\
        &\Big\}T_{\mathcal{H}+\mathcal{D}}(r,s) du dr \nonumber \\
        &= \int_{s}^{t}\int_{s}^{r}T_{\mathcal{H}}^{-}\left(\frac{r+s}{2},
        s\right)\Big\{\nonumber \\
        & \int_{s}^{u}T_{\mathcal{D}}^{-}(u,s)\left(\frac{\mathcal{H}}{2} +
        \mathcal{D}(u)\right)
        T_{\mathcal{H}}^{-}\left(v,\frac{v+s}{2}\right)\left[\frac{\mathcal{H}}{2},
        \mathcal{D}(r)\right] 
        T_{\mathcal{H}}^{-}\left(\frac{u+s}{2},\frac{v+s}{2}\right)dv \nonumber \\
        &+ \int_{s}^{u}
        T_{\mathcal{D}}^{-}(v,s)\left[\mathcal{D}(v),\frac{\mathcal{H}}{2}\right]
        T_{\mathcal{D}}^{-}(u,v) \left(\frac{\mathcal{H}}{2} +
        \mathcal{D}(u)\right)
        T_{\mathcal{H}}^{-}\left(u,\frac{u+s}{2}\right) dv\nonumber \\
        &+ \int_{r}^{u}T_{\mathcal{D}}^{-}(u,s)\left[\frac{\mathcal{H}}{2},
        \partial_{v}\mathcal{D}(v)\right]
        T_{\mathcal{H}}^{-}\left(u,\frac{u+s}{2}\right)dv\nonumber \\
        &\Big\}T_{\mathcal{H}+\mathcal{D}}(r,s) du dr.
    \end{align}
    Returning to equation~\eqref{eq:redefs} we find,
    \begin{align}
        \epsilon &\leq \int_{s}^{t}\int_{s}^{r}\dnorm{T_{\mathcal{H}}\left(t,\frac{t+s}{2}\right)
        T_{\mathcal{D}}\left(t,s\right)
        T_{\mathcal{H}}\left(\frac{t+s}{2},\frac{r+s}{2}\right)} \Big\{\nonumber \\
        & \int_{s}^{u} \dnorm{T_{\mathcal{D}}^{-}(u,s)\left(\frac{\mathcal{H}}{2} +
        \mathcal{D}(u)\right)
        T_{\mathcal{H}}^{-}\left(v,\frac{v+s}{2}\right)\left[\frac{\mathcal{H}}{2},
        \mathcal{D}(r)\right] 
        T_{\mathcal{H}}^{-}\left(\frac{u+s}{2},\frac{v+s}{2}\right)}dv \nonumber \\
        &+ \int_{s}^{u}
        \dnorm{T_{\mathcal{D}}^{-}(v,s)\left[\mathcal{D}(v),\frac{\mathcal{H}}{2}\right]
        T_{\mathcal{D}}^{-}(u,v) \left(\frac{\mathcal{H}}{2} +
        \mathcal{D}(u)\right)
        T_{\mathcal{H}}^{-}\left(u,\frac{u+s}{2}\right)} dv\nonumber \\
        &+ \int_{r}^{u}\dnorm{T_{\mathcal{D}}^{-}(u,s)\left[\frac{\mathcal{H}}{2},
        \partial_{v}\mathcal{D}(v)\right]
        T_{\mathcal{H}}^{-}\left(u,\frac{u+s}{2}\right)}dv\nonumber \\
        &\Big\}\dnorm{T_{\mathcal{H}+\mathcal{D}}(r,s)} du dr \nonumber \\
        &\leq \int_{s}^{t}\int_{s}^{r} \Big\{\nonumber \\
        & \int_{s}^{u} \dnorm{T_{\mathcal{D}}^{-}(u,s)}\dnorm{\left(\frac{\mathcal{H}}{2} +
        \mathcal{D}(u)\right)}
        \dnorm{T_{\mathcal{H}}^{-}\left(v,\frac{v+s}{2}\right)}\dnorm{\left[\frac{\mathcal{H}}{2},
        \mathcal{D}(r)\right]}
        \dnorm{T_{\mathcal{H}}^{-}\left(\frac{u+s}{2},\frac{v+s}{2}\right)}dv \nonumber \\
        &+ \int_{s}^{u}
        \dnorm{T_{\mathcal{D}}^{-}(v,s)}\dnorm{\left[\mathcal{D}(v),\frac{\mathcal{H}}{2}\right]}
        \dnorm{T_{\mathcal{D}}^{-}(u,v)} \dnorm{\left(\frac{\mathcal{H}}{2} +
        \mathcal{D}(u)\right)}
        \dnorm{T_{\mathcal{H}}^{-}\left(u,\frac{u+s}{2}\right)} dv\nonumber \\
        &+ \int_{r}^{u}\dnorm{T_{\mathcal{D}}^{-}(u,s)}\dnorm{\left[\frac{\mathcal{H}}{2},
        \partial_{v}\mathcal{D}(v)\right]}
        \dnorm{T_{\mathcal{H}}^{-}\left(u,\frac{u+s}{2}\right)}dv\nonumber \\
        &\Big\} du dr,
    \end{align}
    where the second inequality follows from $\dnorm{T_{\mathcal{L}}(t,s)} = 1$
    since $T_{\mathcal{L}}(t,s)$ is a CPTP map. Applying part 2 of
    Lemma~\ref{lm:bte} yields,
    \begin{align}
        \epsilon &\leq \int_{s}^{t}\int_{s}^{r} \Big\{\int_{s}^{u} e^{\frac{\dnorm{\mathcal{H}}}{2}(u - s) +
        \int_{s}^{u}\dnorm{\mathcal{D}(t')}dt'}\dnorm{\left(\frac{\mathcal{H}}{2} +
        \mathcal{D}(u)\right)}
        \dnorm{\left[\frac{\mathcal{H}}{2},
        \mathcal{D}(r)\right]} dv \nonumber \\
        &+ \int_{s}^{u}
        e^{\frac{\dnorm{\mathcal{H}}}{2} (u - s) +
        \int_{s}^{u}\dnorm{\mathcal{D}(t')}dt'}\dnorm{\left[\mathcal{D}(v),\frac{\mathcal{H}}{2}\right]}
        \dnorm{\left(\frac{\mathcal{H}}{2} +
        \mathcal{D}(u)\right)} dv\nonumber \\
        &+ \int_{r}^{u} e^{\frac{\dnorm{\mathcal{H}}}{2}(u-s) +
            \int_{s}^{u}\dnorm{\mathcal{D}(t')} dt' }\dnorm{\left[\frac{\mathcal{H}}{2},
        \partial_{v}\mathcal{D}(v)\right]} dv \Big\} du dr \nonumber \\
        &\leq \int_{s}^{t}\int_{s}^{r} e^{\left(\frac{\dnorm{\mathcal{H}}}{2} +
        \sup_{t}\dnorm{\mathcal{D}(t)}\right)(u -
        s)} \left\{\int_{s}^{u}
        \dnorm{\left(\frac{\mathcal{H}}{2} +
        \mathcal{D}(u)\right)}
        \dnorm{\left[\frac{\mathcal{H}}{2},
        \mathcal{D}(r)\right]} dv \right. \nonumber \\
        &+ \int_{s}^{u} \dnorm{\left[\mathcal{D}(v),\frac{\mathcal{H}}{2}\right]}
        \dnorm{\left(\frac{\mathcal{H}}{2} +
        \mathcal{D}(u)\right)} dv\nonumber \\
        &\left. + \int_{r}^{u} \dnorm{\left[\frac{\mathcal{H}}{2},
        \partial_{v}\mathcal{D}(v)\right]} dv \right\} du dr \nonumber \\
        &\leq \int_{s}^{t}\int_{s}^{r} e^{\left(\frac{\dnorm{\mathcal{H}}}{2} +
        \sup_{t}\dnorm{\mathcal{D}(t)}\right)(u -
        s)} \left\{\left(\frac{\dnorm{\mathcal{H}}}{2} +
        \sup_{t}\dnorm{\mathcal{D}(t)}\right)
        \sup_{t}\dnorm{\left[\mathcal{H},\mathcal{D}(t)\right]} (u -
        s)\right. \nonumber \\
        &\left. + \sup_{t}\dnorm{\left[\frac{\mathcal{H}}{2},
        \partial_{t}\mathcal{D}(t)\right]} \left((u - s) - (r - s)\right)
        \right\} du dr.
    \end{align}
    Defining $L = \frac{\dnorm{\mathcal{H}}}{2} +
    \sup_{t}\dnorm{\mathcal{D}(t)}$, $C = \sup_{t}\dnorm{\left[\mathcal{H},
    \mathcal{D}(t)\right]}$, and $D =
    \sup_{t}\dnorm{\left[\mathcal{H},
    \partial_{t}\mathcal{D}(t)\right]}$ gives us the bound,
    \begin{align}
        \epsilon &\leq \int_{s}^{t}\int_{s}^{r} e^{L(u -
        s)} \left\{\frac{2LC + D}{2} (u - s) - \frac{D}{2} (r - s)
        \right\} du dr\nonumber \\
        &\leq \int_{s}^{t} \left\{\frac{2LC + D}{2L^{2}}
        \left[e^{L(r-s)}(L(r-s) - 1) + 1\right] - \frac{D}{2L^{2}} L(r -
        s)\left[e^{L(r-s)} - 1\right]
        \right\} dr \\
        &\leq \left\{\frac{2LC + D}{2L^{3}}
        \left(2\left(1 - e^{L(t-s)}\right) + L(t-s)\left(1 + e^{L(t-s)}\right)
        \right) - \frac{D}{2L^{3}}
        \left((L(t-s) - 1)e^{L(t-s)} + 1 - \frac{\left(L(t -
        s)\right)^{2}}{2}\right) \nonumber 
        \right\}.
    \end{align}
    Rewriting the bound we find,
    \begin{align}
        \epsilon &\leq \left\{\frac{2LC + D}{2L^{3}}
        \left(2\left(1 - e^{L(t-s)}\right) 
        + L (t-s)\left(1 + e^{L(t-s)}\right)
        \right) - \frac{D}{2L^{3}}
        \left((L(t-s) - 1)e^{L(t-s)} + 1 - \frac{\left(L(t-s)
        \right)^{2}}{2}\right)
        \right\} \nonumber \\
        \epsilon &\leq \left\{\frac{2LC + D}{2L^{3}}
        \left(\frac{(L(t-s))^{3}}{3} + \mathcal{O}\left((L(t-s)
        )^{4}\right)\right) - \frac{D}{2L^{3}}
        \left(\frac{(L(t-s))^{3}}{3} + \mathcal{O}\left((L(t-s))^{4}\right)\right)
        \right\}
    \end{align}
    Assuming $\left(\frac{\dnorm{\mathcal{H}}}{2} +
    \sup_{t}\dnorm{\mathcal{D}(t)}\right)(t-s) = L(t-s) \leq 1$, gives us 
    \begin{equation}
        \epsilon \leq \frac{LC}{3}(t-s)^{3} =
        \frac{1}{3}\sup_{t}\dnorm{\left[ \mathcal{H},
        \mathcal{D}(t)\right]}\left(\frac{\dnorm{\mathcal{H}}}{2} +
        \sup_{t}\dnorm{\mathcal{D}(t)}\right)(t-s)^{3}
    \end{equation}
\end{proof}

\section{Evolution generated by unitary jump operators}
Let $\mathcal{D}$ be a Lindbladian dissipator with jump operators $L_{\mu}$.
Setting $L_{\mu} = \alpha_{\mu} U_{\mu}$ where $U_{\mu}$ is unitary gives us 
$\mathcal{D}(\rho) = \sum_{\mu=1}^{m}\left(L_{\mu}\rho
L_{\mu}^{\dagger} - \frac{1}{2}\{L_{\mu}^{\dagger}L_{\mu},\rho\}\right) =
\sum_{\mu=1}^{m} |\alpha_{\mu}|^{2} U_{\mu} \rho U_{\mu}^\dagger-\sum_{\mu=1}^{m}
|\alpha_{\mu}|^{2} \rho$. Define the random unitary channel $\mathcal{R}
(\rho)=\sum_{\mu=1}^{m}
|\alpha_{\mu}|^{2} U_{\mu} \rho U_{\mu}^\dagger$ and
$a=\sum_{\mu=1}^{m}|\alpha_{\mu}|^2$, hence $\mathcal{D}(\rho)= \mathcal{R}-a \mathcal{I}$. Note that  $e^{\delta t \mathcal{D}}=e^{-a \delta t \mathcal{I}} e^{\delta t  \mathcal{R}}$. Assuming $ \delta t\dnorm{\mathcal{R}} \leq 1$, we have
\begin{equation}
    \dnorm{e^{\delta t \mathcal{D}}-e^{- a\delta t }\sum_{k=0}^K{  \frac{\delta t^k}{k!} \mathcal{R}^k (\rho)}} \leq c e^{-a \delta t } \frac{(\delta t\dnorm{\mathcal{R}})^{K+1}}{(K+1)!}.
\end{equation}
Clearly $\sum_{k=0}^K{ e^{-a \delta t} \frac{\delta t^k}{k!} \mathcal{R}^k
(\rho)}$ is a random unitary channel with non-negative coefficients for any
$K$. Therefore the evolution can be approximated with a random unitary channel up to arbitrary accuracy by increasing $K$.
Note that the sum can involve terms that are up to $K$ multiplication of
unitary operators $\{U_{\mu}\}$.  When $\{U_{\mu}\}$ belongs to a group of
efficiently implementable unitaries, such as $n$-fold tensor products of single-qubit unitaries or Clifford circuits which are closed under multiplication, the complexity of implementing the resulting unitary operators does not grow with $K$.

\section{Costs for Simulating $\mathcal{L}_{\gamma}$ and $\mathcal{L}_{\Gamma}$}

Consider the case of a Lindbladian $\mathcal{L}_{\gamma}$ which generates
arbitrary Hamiltonian evolution along with global depolarizing noise,
\begin{align}
    \mathcal{L}_{\gamma}(\rho) &= \mathcal{H}(\rho) + \mathcal{D}_{\gamma}(\rho)
    = -i [H, \rho] + \sum_{\mu=1}^{m}\left(L_{\mu}\rho L_{\mu}^{\dagger} -
        \frac{1}{2}\{L_{\mu}^{\dagger}L_{\mu}, \rho\}\right) \nonumber \\
    &= -i [H, \rho] + \sum_{\mu=1}^{4^{n} -
        1}\frac{\gamma}{4^{n}}\left(V_{\mu}\rho V_{\mu}^{\dagger} -
        \frac{\gamma}{4^{n}}\rho\right)
\end{align}
where $H$ is an arbitrary Hamiltonian, $\gamma \in \mathbb{R}_{+}$, and
$V_{\mu} \in \{\mathbb{I}, X, Y, Z\}^{\otimes n} / \{\mathbb{I}^{\otimes n}\}$.
There are $m = 4^{n} - 1$ jump
operators for $\mathcal{D}_{\gamma}$ of the form
$L_{\mu} = \sqrt{\frac{\gamma}{4^{n}}}V_{\mu}$. Since $V_{\mu}$ is unitary,
$e^{\delta t \mathcal{D}}$ is  approximately stochastically simulatable
up to arbitrary precision as discussed in the above section. Thus oracles
$\mathcal{A}_{\delta t}^{(i)}$ are compositions of elements in $\{\mathbb{I}, X
, Y, Z\}^{\otimes n}/\{\mathbb{I}^{\otimes n}\}$. Calculating
$\pnorm{\mathcal{L}}$ gives
\begin{equation}
    \pnorm{\mathcal{L}} = \pnorm{\mathcal{H}} + \pnorm{\mathcal{D}_{\gamma}}
    = \pnorm{\mathcal{H}} + \sum_{\mu = 1}^{m}\left(\sum_{k=0}^{s-1}\beta_{\mu
    k}\right)^{2}.
\end{equation}
Since jump operators $L_{\mu} = \sqrt{\frac{\gamma}{4^{n}}}V_{\mu}$ are already
in the form of $L_{\mu} = \sum_{k=0}^{s-1}\beta_{\mu k }V_{\mu k}$ where
$\beta_{\mu k} > 0$ and $V_{\mu k}$ is an $n$-fold tensor product of Pauli
operators with arbitrary phase, we find 
\begin{equation}
    \sum_{k=0}^{s-1}\beta_{\mu k} = \sqrt{\frac{\gamma}{4^{n}}}.
\end{equation}
This implies 
\begin{equation}
    \pnorm{\mathcal{L}}
    = \pnorm{\mathcal{H}} + \sum_{\mu = 1}^{4^{n} - 1}\frac{\gamma}{4^{n}} =
    \pnorm{\mathcal{H}} + \left(1 - \frac{1}{4^{n}}\right)\gamma =
    O\left(\pnorm{\mathcal{H}} + \gamma \right).
\end{equation}

Since oracles $\mathcal{A}_{\delta t}^{(i)}$ are compositions of elements in
$\{\mathbb{I}, X, Y, Z\}^{\otimes n} / \{\mathbb{I}^{\otimes n}\}$,
$\mathcal{A}_{\delta t}^{(i)}$ are $n$-fold tensor products of single-qubit
Pauli operators with arbitrary phase. Thus each
$\mathcal{A}_{\delta t}^{(i)}$ requires $O(n)$ single-qubit gates to implement. Applying
Theorem~\ref{thm:costs} gives us the total cost to simulate the dynamics
described by $\mathcal{L}_{\gamma}$ to be 
\begin{equation}
O\left(n  \frac{((\gamma +
\pnorm{\mathcal{H}})T)^{1.5}}{\sqrt{\epsilon}}\right)
\end{equation}
1- and 2-qubit gates along with 
\begin{equation}
    \Tilde{O}\left(\log(T / \epsilon)\right)
\end{equation}
ancilla if choosing to use the Hamiltonian simulation subroutine of
\cite{berry2015simulating}. 

As another example, consider the case when $\mathcal{L}_{\Gamma}$
generates 1-local dephasing noise. In this case,
the jump
operators have the form $L_{\mu} = \sqrt{\frac{\Gamma}{2}}Z_{\mu}$ where
$Z_{\mu}$ acts on the $\mu$-th qubit. The analysis follows similar to above
with  $\gamma$ replaced by $n\Gamma / 2$. Similarly, oracles
$\mathcal{A}_{\delta t}^{(i)}$ become  $n$-fold tensor products of single qubit
unitaries.

\end{document}